\documentclass{sig-alternate-05-2015}

\usepackage[applemac]{inputenc}
\usepackage[T1]{fontenc}
\usepackage{amsmath,amsfonts, amssymb, url}
\usepackage[english]{babel}
\usepackage[plainpages=false,pdfpagelabels,colorlinks=true,citecolor=blue,hypertexnames=false]{hyperref}
\usepackage[algoruled,vlined,english,linesnumbered]{algorithm2e}
\usepackage[all]{xy}
\usepackage{array,multirow}
\usepackage{enumitem}
\setlist{leftmargin=*}
\usepackage{tikz}

\usepackage{amsthm}

\definecolor{bleufonce}{rgb}{0,0,0.4}

\newcommand{\R}{\mathbf{R}}

\newcommand{\Z}{\mathbf{Z}}
\newcommand{\Q}{\mathbf{Q}}
\newcommand{\F}{\mathbf{F}}

\newcommand{\End}{\text{\rm End}}

\newcommand{\SM}{\text{\rm SM}}
\newcommand{\matw}{\text{Mat}_{\text{work}}}

{\theoremstyle{plain}
\newtheorem{thm}{Theorem}[section]}
\newtheorem{lem}[thm]{Lemma}
\newtheorem{prop}[thm]{Proposition}
\newtheorem{cor}[thm]{Corollary}

{\theoremstyle{definition}

\newtheorem{rmq}[thm]{Remark}
\newtheorem{defi}[thm]{Definition}
\newtheorem{notation}[thm]{Notation}
\newcommand{\gal}{\textrm{Gal}}

\begin{document}
\title{Fast multiplication for skew polynomials}
\numberofauthors{2}
\author{
\alignauthor
Xavier Caruso\\
       \affaddr{IRMAR, CNRS\\
       Campus de Beaulieu\\
       263 avenue du G{\'e}n{\'e}ral Leclerc\\
       35042  RENNES Cedex}\\
       \email{xavier.caruso@normalesup.org}
\alignauthor
J{\'e}r{\'e}my Le Borgne\\
       \affaddr{IRMAR, ENS Rennes, UBL\\
       Campus de Ker Lann\\
       Avenue Robert Schuman\\
       35170 BRUZ}\\
       \email{jeremy.leborgne@ens-rennes.fr}
}
\maketitle
\begin{abstract}
We describe an algorithm for fast multiplication of skew polynomials. 
It is based on fast modular multiplication of such skew polynomials, for 
which we give an algorithm relying on evaluation and interpolation on 
normal bases. Our algorithms improve the best known complexity for these 
problems, and reach the optimal asymptotic complexity bound for large 
degree. We also give an adaptation of our algorithm for polynomials of 
small degree. Finally, we use our methods to improve on the best known 
complexities for various arithmetics problems.
\end{abstract}

\section*{Introduction}


The present paper is dedicated to the description of algorithms for fast arithmetics in skew polynomial rings. Since they were first introduced by Ore, skew polynomials and their variants have been widely studied in several areas of mathematics. In particular, skew polynomials over finite fields have various applications in coding theory~\cite{SiKs09}, 
cryptography~see \cite{BuHe14}, for $p$-adic Galois representations~\cite{LeB}. Fast arithmetics for manipulating these objects is useful for such applications, and has been improved over time since the first breakthrough paper on computational skew polynomials over finite fields, due to Giesbrecht \cite{Gi98}.

Let $K$ be a field and let $L$ be a finite extension  of $K$, endowed with the endomorphism $\sigma$. We assume that $\sigma$ has order $r \geq 1$ and that $K = L^\sigma$. We consider the ring $L[X,\sigma]$ of skew polynomials with coefficients in $L$. This is a non commutative ring where the relation $Xa = \sigma(a) X$ holds for all $a \in L$ (for more detail about the definitions, see section \ref{subs:defi}). The main problem addressed in this paper is the fast multiplication of elements of $L[X,\sigma]$. The complexity of algorithms is described in terms of the number of elementary operations in $K$ with respect to the degree $d$ of the skew polynomials to be multiplied, and the degree $r$ of $L$ over $K$.

\smallskip

\noindent
\textbf{State of the art.}
The na{\"i}ve method for multiplication of skew polynomials of degree $\leq d$ yields an algorithm that has complexity $O(d^2r^2)$ operations in $K$. In \cite{Gi98}, this complexity is improved to $O(dr^2 + d^2r)$. Let $\omega$ denote the exponent of matrix multiplication. The authors of the present paper gave several algorithms for multiplication in \cite{CaLe17}, with best complexity $\tilde O(dr^{\omega-1})$ achieved for $d> r^2$. The most recent results by Puchinger and Wachter-Zeh \cite{PuWa16} give a bound of $\tilde O(d^\frac{\omega+1}{2} r)$ operations in $K$ for multiplication in $L[X,\sigma]$, which improves on the previous results \cite{CaLe17} when $d \in \Theta(r)$, which is the most relevant case for applications in coding theory (see \cite{PuWa16}, \S 4.2). 
In the context of differential operators (which share many similarities with skew polynomials),
Benoit, Bostan and Van der Hoeven have obtained a complexity of $\tilde{O}(\min\{d,r\}^{\omega-2}dr)$ (see \cite{BeBoVa12}, Theorem 1) for multiplication in $L[x]\langle\partial\rangle$. We expect that this complexity should be doable in $L[X,\sigma]$ as well, but we have only achieved it for $d \geq r$.

\smallskip

\noindent
\textbf{Contributions of the paper.}
This paper's main algorithm improves the complexity of the best known algorithms for multiplication in $L[X,\sigma]$ 
to $\tilde O(dr^{\omega-1})$ when $d \geq r$. For $d \in \Theta(r)$, this gives a complexity of $\tilde O(r^\omega)$ operations in $K$. This is quasi-optimal in the sense that matrix multiplication can be reduced to skew polynomial multiplication (this is for example a consequence of Proposition~\ref{prop:evalbasis} below), so that any improvement on the exponent of skew polynomial multiplication would lead to a similar improvement for matrix multiplication. We also design a new algorithm for multiplication of polynomials of small degree $d \ll r$ in $L[X,\sigma]$, whose complexity is $\tilde O(d^{\omega-2}r^2)$.


We also show that our method can be used to improve the best known complexities of various related problems, such as multi-point evaluation, minimal subspace polynomial, and interpolation which are studied in \cite{PuWa16}. We also improve the complexities for greatest common divisors and least common multiples.

\smallskip

\noindent
\textbf{Organization of the paper.}
The first section of the paper focuses on elementary operations for skew 
polynomials with normal bases: evaluation and interpolation. More 
precisely, if $P \in L[X,\sigma]$, then $P(\sigma)$ is an endomorphism 
of the $K$-algebra $L$, and the map $P \mapsto P(\sigma)$ is a morphism 
of $K$-algebras. In this section, we describe how we can compute 
efficiently $P(\sigma)$ using a normal basis and, conversely, how to 
recover $P$ (the reduction modulo $X^r {-} 1$ of) $P$ from the datum of 
$P(\sigma)$ (see Proposition \ref{prop:evalbasis}). We also look into 
more detail how the can solve the same evaluation/interpolation problems 
with $P$ of small degree $n$ at only the first $n$ elements of a normal 
basis.

In the second section, we present our algorithm for fast multiplication of skew polynomials. First, we study how the multiplication can be done efficiently modulo $X^r {-} 1$ through evaluation/interpolation on a normal basis and matrix multiplication. We generalize this study to multiplication modulo $Z(X^r)$ for any irreducible polynomial $Z \in K[T]$. This allows us to give an algorithm for multiplication of skew polynomials of degree $d$ that works in $O(dr^{\omega-1})$ operations in $K$ (where $r^\omega$ denotes the complexity of multiplication of square matrices of size $r$).

In the third section, we give several applications to fast arithmetics for skew polynomials. We show how we can perform general multi-point evaluation, minimal subspace polynomial, and interpolation, as well as usual operations on skew polynomials such as (extended) Euclidean division, greatest common divisor, least common multiple.

\section{Fast evaluation\\and interpolation} 
In this section, we present the notion of skew polynomials, and we study the problems of their evaluation and interpolation using normal bases.

\subsection{Definitions and notations}\label{subs:defi}
Let $K$ be a field and let $L$ be an {\'e}tale $K$-algebra (since $K$ is a field, this just means that $L$ is isomorphic to a product of field extensions of $K$). Let $\sigma$ be an automorphism of $L$. We assume that $\sigma$ has finite order $r$ and that $K = L^\sigma$. The ring $L[X,\sigma]$ of skew polynomials with coefficients in $K$ is the ring whose underlying group is $L[X]$ and whose multiplication is determined by the relation
$$ \forall \alpha \in L,~ X\alpha = \sigma(\alpha)X.$$
The ring $L[X,\sigma]$ is not commutative unless $r=1$.

\smallskip

\noindent
\textbf{Examples.}
The following situations are examples of the general setting that we are considering:
\begin{itemize}
\renewcommand{\itemsep}{0pt}
\item $L = K^r$, and $\sigma$ is the shift operator $(x_0, \dots, x_{r-1}) \mapsto (x_1, \dots, x_{r-1}, x_0)$,
\item (Extensions of finite fields) $K = \F_q$, $L = \F_{q^r}$ and $\sigma : x \mapsto x^q$ is the Frobenius endomorphism of $L$,
\item (Cyclotomic extensions) $K = \Q$ and $L = \Q(\zeta_{p^n})$ where $\zeta_{p^n}$ is a primitive $p^n$-th root of unity and $p$ is prime; $\sigma$ is a generator of the Galois group $\gal(L/K)$ (which is the cyclic group $(\Z/p^n\Z)^\times$).
\item (Kummer extensions) $K$ contains a primitive $r$-th root $\zeta_r$ of $1$, $L = K(\sqrt[r] a)$ for some suitable $a \in K$ and $\sigma$ takes $\sqrt[r] a$ to $\zeta_r \sqrt[r] a$.
\end{itemize}
The two last examples are addressed in~\cite{Ro15} and have applications
to space-time codes.

\begin{rmq}
Usually, $L$ is assumed to be a field extension of $K$. We are considering the more general context of an {\'e}tale $K$-algebra because it is stable under base change: if $L/K$ is {\'e}tale and $K'$ is an extension of $K$, then $L' = L \otimes_K K'$ is {\'e}tale over $K'$ (but it is not a field in general, even if $L$ is). This feature is used mostly in Section \ref{subs:modmultZ}, and does not make the classical results any more difficult to prove.
\end{rmq}
\begin{defi}
A \emph{normal basis} of $L/K$ is a basis $(b_0, \dots, b_{r-1})$ of $L$ over $K$ such that $\sigma(b_{i+1}) = b_i$ (the indices being taken modulo $r$).
\end{defi}
\begin{prop}[\cite{De32}, Satz 1]
Assuming $\sigma$ has order $r$ and $K = L^\sigma$, $L$ has a normal basis.
\end{prop}
The problem of the construction of normal bases has been widely studied, see for example \cite{GaGi90} for the case of finite fields, and \cite{Gi99} for the case of number fields.
In both cases of cyclotomic extensions and Kummer extensions, it is easy to
exhibit a normal basis: in the cyclotomic case, the basis starting with
$b_0 = \zeta_{p^n}$ does the job while in the Kummer case, one can take:
$$b_0 = 
1 + \sqrt[r]{a} + 
\sqrt[r]{a^2} + \cdots + \sqrt[r]{a^{r-1}} = \frac{a-1}{\sqrt[r]{a}-1}.$$
From now on, we assume that we have fixed a normal basis $(b_0, \dots, 
b_{r-1})$ of $L$ together with a working basis in which the elements of $L$
are represented.
Let $\Omega$ be the matrix of change of basis from the working basis to 
the normal basis. We assume that the multiplication in $L$ and the 
application of $\sigma$ can be both performed in $\tilde O(r)$ 
operations in $K$ in the working basis.

\subsection{Evaluation and interpolation\\on a normal basis}\label{sec:eval-inter}
We introduce a relation between polynomials that allows to evaluate the linear map associated to a skew polynomial at the elements of the normal basis $(b_0, \dots, b_{r-1})$. 
\begin{lem}\label{lem:iso1}
The map:
$$ 
\begin{array}{rcl}
L[X,\sigma] & \rightarrow & \End_K(L)\\
A = \sum_{i\geq 0}a_iX^i& \mapsto & A(\sigma) = \sum_{i \geq 0} a_i \sigma^i
\end{array}
$$
is a homomorphism of $K$-algebras. It induces an isomorphism of $K$-algebras:
$$ \varepsilon ~:~L[X,\sigma]/(X^r {-} 1) \simeq \End_K(L).$$
\end{lem}
\begin{proof}
The first map is a homomorphism because for all $a \in L$, $Xa = \sigma(a)X$ in $L[X, \sigma]$. Since $\sigma$ has order $r$, $X^r {-} 1$ lies in the kernel of this map, so $\varepsilon$ is well-defined. Both $L[X,\sigma]/(X^r {-} 1)$ and $\End_K(L)$ are $K$-vector spaces of dimension $r^2$, hence it suffices to prove injectivity. By Artin's Lemma on independence of characters, $\{id, \sigma, \ldots, \sigma^{r-1}\}$ is a linearly independent family over $L$, so that if $P(\sigma) = 0$ for some $P\in L[X,\sigma]$ of degree $<r$, then $P = 0$.
\end{proof}
Lemma \ref{lem:iso1} shows that multiplication of skew polynomials modulo $X^r {-} 1$ is essentially the same as multiplication of $r\times r$ matrices over $K$, assuming that the isomorphism $\varepsilon$ can be computed efficiently (in both ways). We now address this question.
\begin{notation}
Throughout this paper, we will denote $P(x)$ for $P(\sigma)(x) = \varepsilon(P)(x)$ if $P \in L[X,\sigma]$ and $x \in L$.
\end{notation}


Let $T$ be a new (commutative) variable and consider the classical polynomial ring $L[T]$. Let $B = \sum_{i=0}^{r-1} b_i T^i \in L[T]$ be the polynomial whose coefficients are the elements of the normal basis.

\begin{prop}\label{prop:evalbasis}
Let $A = \sum_{i=0}^{r-1}a_iX^i \in L[X,\sigma]$ and let $\tilde A(T) = \sum a_i T^i \in L[T]$. Let $c_j = A(b_j)$ and let $C(T) = \sum_{j=0}^{r-1}c_jT^j$. Then 
$$C(T) = \tilde A(T)B(T) \pmod {T^r {-} 1}.$$
\end{prop}
\begin{proof}
By linearity, it is enough to check that the relation holds when $A = X^i$ for $0 \leq i \leq r-1$. Let $0 \leq i \leq r-1$. We have $X^i(b_j) = \sigma^i(b_j) = b_{j-i}$, where indices are taken modulo $r$.\\
On the other hand, doing the calculations modulo $T^r {-} 1$, $ T^iB(T) = \sum_{j=0}^{r-1} b_{j-i}T^j$.
\end{proof}
Proposition \ref{prop:evalbasis}, although elementary, shows that the isomorphism $\varepsilon$ of Lemma \ref{lem:iso1} can be computed efficiently. Moreover, it also shows how the inverse isomorphism can be computed. More precisely:
\begin{cor}\label{cor:complexity_mod1}
Multiplication in $L[X,\sigma]/(X^r {-} 1)$ can be performed in $O(r^\omega)$ operations in $K$.
\end{cor}
\begin{proof}

Let $A_1,A_2 \in L[X,\sigma]/(X^r {-} 1)$. Let $\tilde A_1(T), \tilde 
A_2(T) \in L[T]$ be the commutative polynomials with the same 
coefficients as $A_1, A_2$ respectively. Let $C_1(T) = \tilde A_1(T)B(T) 
\in L[T]/(T^r {-} 1)$ and $C_2(T) = \tilde A_2(T)B(T) \in L[T]/(T^r {-} 
1)$. Both $C_1$ and $C_2$ can be computed in $\tilde O(r^2)$ operations 
in $K$. Now let $M_1$ (resp. $M_2$) be the matrix whose $j$-th column is 
the decomposition of the $j$-th coefficient of $C_1$ (resp. $C_2$) in 
the working basis. By Proposition \ref{prop:evalbasis}, 
$M_1$ (resp $M_2$) is the matrix of $\varepsilon(A_1)$ (resp. 
$\varepsilon(A_2)$) where the codomain in endowed with the normal basis
and the codomain is endowed with the working basis. Set $M = M_1
\Omega M_2$; this product
can be computed within $O(r^\omega)$ operations in $K$. We know that $M$ is the 
matrix of $\varepsilon(A_1A_2)$ where again
the codomain in endowed with the normal basis
and the codomain is endowed with the working basis.
Let
$$C(T) = \begin{pmatrix} b_0 & b_1 & \cdots & 
b_{r-1}\end{pmatrix}M\begin{pmatrix} 1 \\ T \\ \vdots \\ 
T^{r-1}\end{pmatrix},$$
and compute $\tilde A(T) = C(T)B(T)^{-1} \pmod{T^r {-} 1} = 
\sum_{i=0}^{r-1} a_i T^i$, which can also be computed in $\tilde O(r^2)$ 
operations in $K$. Then, again by Proposition \ref{prop:evalbasis}, 
$A_1A_2 = \sum_{i=1}^{r-1} a_i X^i$. This shows that the global 
complexity of this computation is $O(r^\omega)$.
\end{proof}

In Section \ref{sec:fastmult}, we will generalize this algorithm and show how it yields a fast multiplication algorithm for skew polynomials (not only in the modular case).

 

\subsection{Evaluation and interpolation at\\
 an incomplete normal basis}
\label{subs:evalincnorm}\label{subs:interp_small}

\smallskip

\noindent
\textbf{Evaluation.}
We shall see later how we can compute the product of two skew 
polynomials of small degree $d$ by determining how their product acts on 
$2d$ elements of a normal basis. With this motivation in mind, let us 
describe how we can compute efficiently the image of the first few 
elements of a normal basis under the action of the skew polynomial $A 
\in L[X,\sigma]$. Recall that, using Proposition \ref{prop:evalbasis} 
with $\lambda = 1$, and writing $B(T) = \sum_{i=0}^{r-1} b_iT^i$, we 
know that
$$ \tilde A(T)B(T) \equiv C(T) \pmod{T^r {-} 1},$$
where $C(T) = \sum_{i=0}^{r-1} A(b_i)T^i$. Let $n <r$, and let $A \in L[X,\sigma]$ of degree $n$. We are interested in computing only $A(b_i)$ for $0 \leq i \leq n-1$.
\begin{lem}\label{lem:firstcoeffs}
Let $A \in L[X,\sigma]$ of degree $n$ and let $c_i = A(b_i)$ for $0 \leq i \leq n-1$. Let $U(T) = \sum_{i=0}^n b_iT^i$ and $V(T)  = \sum_{i=0}^n b_{r-i}T^{r-i}$. Then, for $0 \leq i \leq n$:
$$ c_i = \gamma_i + \tilde \gamma_i,$$
where $\tilde A(T)U(T)= \sum_{i=0}^{r-1} \gamma_iT^i$ and $\tilde A(T)V(T) = \sum_{i=0}^{r-1} \tilde\gamma_iT^i$ (the products being taken modulo $T^r {-} 1$).
\end{lem}
\begin{proof}
Since $c_i = \sum_{j+j' = i \pmod r} a_{j}b_{j'}$, and $a_j = 0$ for $j>n$, we are left with the formula:
$$ c_i = \sum_{j' = 0}^{i} a_{i-j'}b_{j'} +\sum_{j' = r-i}^{r-1} a_{i-j' +r} b_{j'},$$
and both sums correspond precisely to the coefficients of $\tilde AU$ and $\tilde AV$ respectively.
\end{proof}
\begin{cor}\label{cor:eval_small}
Let $A \in L[X,\sigma]$ of degree $\leq n$, then the collection of $A(b_0), \ldots, A(b_{n-1})$ can be computed in $\tilde{O}(rn)$ operations in $K$.
\end{cor}
\begin{proof}
By Lemma \ref{lem:firstcoeffs}, the evaluation of $A$ at $b_0, \dots, b_{n-1}$ can be obtained by two multiplications of (classical) polynomials of degree $n$ with coefficients in $L$, hence with complexity $\tilde{O}(nr)$ operations in $K$.
\end{proof}


\noindent
\textbf{Interpolation.}
Still bearing in mind the aim of multiplying two skew polynomials by composing the corresponding linear maps, we are interested in the following question of interpolation: given $n$ values $\alpha_0, \dots, \alpha_{n-1} \in L$, find $A \in L[X,\sigma]$ of degree $n$ such that $A(b_i) = \alpha_i$ for all $0\leq i \leq n-1$.

Let us explain first how the solution to this problem can be computed when $\alpha_0 = \cdots = \alpha_{n-1} = 0$. In this case, the skew polynomial we are looking for is the so-called \emph{minimal subspace polynomial} corresponding to the span $\left< b_{0}, \dots, b_{n-1}\right>$. 
A generic fast algorithm for solving this problem has been proposed by
Puchinger and Wachter-Zeh in~\cite{PuWa16}, Theorem 26; it has complexity
$\tilde O(n^{\max\{\log_2(3), \frac{\omega+1}{2}\}}r)$ operations in $K$.
In the special case we are considering, we shall see that this can be 
improved to $\tilde{O}(nr)$.

Let $B_n(T) = \sum_{i=0}^{r-1} b_{i+n}T^i$, so that $B_n(T) \equiv T^{-n}B(T) \pmod{T^r {-} 1}$. If $A$ is such that $A(b_i) = 0$ for $0 \leq i \leq n-1$, then there exists $Q \in L[T]$ of degree $\leq r{-}n{-}1$ such that $\tilde A(T)B(T) \equiv T^n Q(T) \pmod{T^r {-} 1}$. Of course, the converse is also true, and this equation is equivalent to:
$$ \tilde A(T)B_n(T) \equiv Q(T) \pmod{T^r {-} 1},$$ 
with $\deg \tilde A \leq n$ and $\deg Q \leq r-n-1$. The latter equation can be solved thanks to the extended Euclidean algorithm. Indeed, computing the gcd of $T^r {-} 1$ and $B_n(T)$ and stopping after the first remainder of degree $<r-i$, we get a relation of the form:
$$ U_i(T)B_n(T) + V_i(T)(T^r {-} 1) = Q_i(T),$$
with $\deg U_i \leq i$ and $\deg Q_i \leq r-1-i$, which yields a solution to the problem when $i=n$.
This computation can be done in $\tilde O(nr)$ operations in $K$ thanks to the half-gcd algorithm (see \cite{GaGe03}, Theorem 11.5).

\smallskip

In the general case, let $\alpha_0, \ldots, \alpha_{n-1} \in L$, and let $C(T) = \sum_{i=0}^{n-1} \alpha_iT^i$. We are looking for $A \in L[X, \sigma]$ with degree $\leq n$ and $Q \in L[T]$ with degree $\leq r - n - 1$ such that $A(T) B(T) \equiv C(T) + T^n Q(T) \pmod{T^r {-} 1}$. This equation is equivalent to $A(T)B_n(T) \equiv \sum_{i=0}^{n-1} \alpha_{i}T^{r-n+i} \pmod{T^r {-} 1}$.

\begin{lem}\label{lem:normalsequence}
Let $R_0(T) = T^r {-} 1$, $R_1(T) = B(T)$ and for $i \geq 2$, let $R_i$ be the remainder of the Euclidean division of $R_{i-2}$ by $R_{i-1}$. Then for $0 \leq i \leq r$, $\deg R_i = r-i$.
\end{lem}
\begin{proof}
Consider the map 
$$
\begin{array}{rcl}
\varphi_i : \quad 
L[T]_{<i} \times L[T]_{< i-1} &\longrightarrow &L[T]_{<r + i -1}/L[T]_{< r- i}\\
(U,V) & \mapsto & UR_1 + V R_0
\end{array}
$$
It is well-defined, linear, and both sides have the same dimension over $L$. Moreover, the determinant of this map is nonzero if and only if $\deg R_i = r-i$ (see \cite{Wi96}, \S 4.1). 
Therefore, it is sufficient to prove that $\varphi_i$ is injective.

Let us consider $(U,V)$ in the kernel of $\varphi_i$. 
By definition, $\deg (UR_1 + VR_0) < r-i$, so that $U(T)B(T) \equiv W(T) \pmod{T^r {-} 1}$, where $\deg W(T) < r-i$. By Proposition \ref{prop:evalbasis}, the skew polynomial $\sum_{j=0}^{i-1} u_j X^j$ (whose coefficients are the coefficients of $U$) evaluates to $0$ at $b_{r-i}, \dots, b_{r-1}$. Hence, it is a left multiple of the minimal subspace polynomial $M$ of $\left< b_{r-i}, \dots, b_{r-1}\right>$. Since $(b_{r-i}, \dots, b_{r-1})$ is linearly independent over $K$, $M$ has degree $i$ (it is a generator of the kernel of the $K$-linear map $L[X, \sigma] \rightarrow L^i$ mapping $P$ to $(P(b_{r-i}), \dots, P(b_{r-1}))$). In particular, since $\deg U < i$, $U = 0$, so $V = 0$ and $\varphi_i$ is injective. Hence $\det \varphi_i \neq 0$ and $R_i$ has the required degree.
\end{proof}
\begin{thm}\label{thm:interp_small}
Let $n\leq r$ and $\alpha_0, \ldots, \alpha_{n-1} \in L$. Then there exists $U,V, H \in L[T]$, with $\deg U \leq n-1$, $\deg V \leq n$ and $\deg H \leq r - n$ such that
$$ U(T^r {-} 1) + VB(T) = H(T) + T^{r-n +1}(\alpha_0 + \cdots + \alpha_{n-1}T^{n-1}).$$
Moreover, Algorithm \textrm{\tt SmallDegreeInterpolation} outputs $U$ and $V$ 
for a cost of $\tilde O(rn)$ operations in $K$.
\end{thm}
\begin{proof}[Sketch of the proof]
The result follows from the correctness of Algorithm \ref{algo:interp_small}, but is also a theoretical consequence of Lemma \ref{lem:normalsequence}. Indeed, this lemma shows that there exists a linear combination of $R_0 = T^r {-} 1$, $R_1 = B(T), \ldots, R_{n-1}$ whose higher degree terms have coefficients $c_0, \ldots, c_{n-1}$, and the bounds on the degrees follow from the fact that for $i \leq n$, $R_i = U_iR_0 + V_iR_1$ with $\deg U_i \leq i-1$, $\deg V_i \leq i$.
Algorithm \ref{algo:interp_small} is an adaptation of the half-gcd algorithm, which computes simultaneously the sequence of the remainders in the extended Euclidean division or $R_0$ and $R_1$, and the combination of $R_1$ and $R_0$ that has the given higher degree terms.
\end{proof}

\begin{algorithm}[h]\label{algo:interp_small}
\KwIn{$R_0, R_1 \in L[T]$, $a_0, \dots, a_{k-1} \in L$, with $k \leq n_0$}
\KwOut{$M \in L[T]^{2 \time 2}$ such that $M\begin{pmatrix} R_0\\ R_1 \end{pmatrix} = \begin{pmatrix} R_{k-1} \\ R_k \end{pmatrix}$ and 
$N \in L[T]^{1 \times 2}$ such that $N\begin{pmatrix} R_0\\ R_1 \end{pmatrix} = S_k$ with $\deg S_k - (a_{k-1} + a_{k-2}T + \cdots + a_0T^{k-1})T^{n_0-k+1} \leq n_0-k$ }
$h:=\lfloor k/2 \rfloor$\\
$\tilde R_0 = R_0 \text{ quo } T^{2h}$, $\tilde R_1 = R_1 \text{ quo } T^{2h-1}$\\
$M_1, N_1 = \texttt{SmallDegreeInterpolation}(\tilde R_0, \tilde R_1, a_0, \cdots a_{h-1})$\\
$\begin{pmatrix}
R_{h-1}\\R_h 
\end{pmatrix} := M_1\begin{pmatrix}
R_{0}\\R_1 
\end{pmatrix},~
S := N_1\begin{pmatrix}
R_{0}\\R_1 
\end{pmatrix} = \sum_{i=0}^{n_0-h} s_iT^i + \sum_{i=0}^{h-1}a_i T^{n_0-i}$\\
Make the Euclidean divisions:\\
$R_{h-1}  = Q_hR_{h} + R_{h+1}$\\
$R_{h-1} - a_{h}T^{n_0-h} = \tilde Q_h R_h + \tilde R_{h+1}$\\
$M_2, N_2 = \texttt{SmallDegreeInterpolation}(R_{h}, R_{h+1}, a_{h+1}-s_{n_0-h}, \cdots, a_{2h}-s_{n_0-2h+1})$\\
\Return{$M_2\begin{pmatrix} 0& 1\\ 1 & -Q_h \end{pmatrix}M_1, N_1 + \begin{pmatrix} 1 & -\tilde Q_h \end{pmatrix}M_1+ M_2N_2$}
\caption{\texttt{SmallDegreeInterpolation}}
\end{algorithm}
Thanks to Corollary \ref{cor:eval_small}, Theorem \ref{thm:interp_small} and Algorithm \ref{algo:interp_small}, we can solve the problem of evaluation and interpolation at the first $n$ elements of an incomplete normal basis in $\tilde O(nr)$ operations in $K$.

\section{Fast multiplication}\label{sec:fastmult}

In this section, we study the problem of multiplying efficiently two 
elements $A_1, A_2 \in L[X,\sigma]$ both of degree $\leq d$. The 
complexity is the number of operations in $K$, given as a function of 
$d$ and $r = \dim_K L$.

\subsection{Modular multiplication}

\subsubsection{Multiplication modulo $X^r{-}a$} 
\label{subs:fastmod}

We consider the ring $L[X,\sigma]$. Let $\lambda \in L^\times$, and let $a = N_{L/K}(\lambda)$. We are now going to describe an algorithm for multiplication in $L[X,\sigma]$ modulo $X^r {-} a$.
\begin{prop}\label{prop:modmulta}
The map 
$$\begin{array}{rcl}
L[X,\sigma]& \longrightarrow &L[X,\sigma]/(X^r {-} 1)\\
A(X) = \sum a_i X^i & \mapsto & A(\lambda X) = \sum_{i} \lambda\sigma(\lambda)\cdots \sigma^{i-1}(\lambda) a_iX^i
\end{array}
$$
factors as an isomorphism $L[X,\sigma]/(X^r {-} a) \simeq L[X,\sigma]/(X^r {-} 1)$.
\end{prop}
\begin{proof}
This maps $X^r$ to $\lambda\sigma(\lambda)\cdots\sigma^{r-1}(\lambda)X^r = aX^r$, thus mapping $X^r {-} a$ to $a(X^r {-} 1)$.
\end{proof}
\begin{cor}\label{cor:modmulta}
Multiplication in $L[X,\sigma]/(X^r {-} a)$ can be performed in $O(r^\omega)$ operations in $K$.
\end{cor}
\begin{proof}
By Proposition \ref{prop:modmulta} and Proposition \ref{prop:evalbasis}, it is enough to show that for $A \in L[X,\sigma]/(X^r {-} a)$, $A(\lambda X)$ can be computed in $O(r^\omega)$ operations in $K$. 
For this we write $\lambda_i = \lambda \sigma(\lambda)\cdots \sigma^{i-1}(\lambda)$ 
and remark that the $\lambda_i$'s ($0 \leq i < r$) can be all computed 
within $\tilde O(r^2)$ operations in $K$ thanks to the recurrence
formula $\lambda_{i+1} = \lambda \cdot \sigma(\lambda_i)$.
Now evaluating the formula $A(\lambda X) =  \sum_{i} 
\lambda_i a_iX^i$ allows us to compute $A(\lambda X)$
in $\tilde O(r^2)$ operations in $K$.
\end{proof}
We could use the proof of Corollary \ref{cor:modmulta} directly to design an algorithm for multiplication modulo $X^r {-} a$. Such an algorithm would require computing $A_1(\lambda X)$ and $A_2(\lambda X)$ each time we use it to compute $A_1A_2$. Alternatively, we can slightly modify the basis on which we are evaluating the corresponding maps, which can provide a gain if there are many multiplications to do modulo $X^r {-} a$.

\medskip

Let $\lambda \in L^\times$, and let $\sigma_a = \lambda\sigma$. Let $\tilde b_{r-1} \in L$, and for $0 \leq i \leq r-2$, $\tilde b_i = \sigma_a^{r-1-i}(\tilde b_{r-1})$, such that $\tilde{\mathcal B} = (\tilde b_0, \dots, \tilde b_{r-1})$ is a basis of $L$ over $K$. By construction, we have for $1 \leq i \leq r-1$, $\sigma_a(\tilde b_i) = \tilde b_{i-1}$, and $\sigma_a(b_0) = a\tilde b_{r-1}$. For example, if $\mathcal B = (b_0, \ldots, b_{r-1})$ is a normal basis of $L$ over $K$, then $\tilde b_{r-1} = b_{r-1}$ and $\tilde b_i = \lambda \sigma(\lambda)\cdots \sigma^{i-1}(\lambda)b_i$ defines a suitable basis.
Now, let $\tilde B = \sum_{i=0}^{r-1} \tilde b_i T^i \in L[T]$.
\begin{prop}\label{prop:evalbasisa}
Let $A = \sum_{i=0}^{r-1}a_iX^i \in L[X,\sigma]$ and let $\tilde c_j = A(\sigma_a)(\tilde b_j)$. Let $\tilde A(T) = \sum a_i T^i \in L[T]$. Let $\tilde C_a = \sum_{j=0}^{r-1}\tilde c_jT^j$. Then 
$$\tilde C_a(T) = \tilde A(T)\tilde B(T) \pmod {T^r {-} a}.$$
\end{prop}
\begin{proof}
The proof is similar to that of Proposition \ref{prop:evalbasis}. By linearity, it is enough to check that the relation holds for $A = X^i$ for $0 \leq i \leq r-1$. Let $0 \leq i \leq r-1$. We have :
$$\sigma_a^i(b_j) = \left\{\begin{array}{lc} b_{j-i} & \text{if } j\geq i\\ ab_{r+j-i} & \text{if } i > j \end{array}\right. .$$
On the other hand, doing the calculations modulo $T^r {-}a$:
$$T^i B(T) = \sum_{j=0}^{r-1} b_j T^{i+j} = \sum_{j=i}^{r-1} b_{j-i}T^i + \sum_{j=0}^{i-1} ab_{r+j-i} T^{j}.$$
Hence, $T^i B(T) = C_{X^i}(T)$ for all $0 \leq i \leq r-1$, so $C_a(T) = \tilde{A}(T)B(T)$ for all $A \in L[X,\sigma]/(X^r {-} a)$.
\end{proof}

Algorithm \texttt{ModMult} below makes precise the algorithmical content 
of Proposition~\ref{prop:evalbasisa}; it uses a primitive $\matw$ that 
takes as input a tuple $(x_1, \ldots, x_r) \in L^r$ and outputs the $r 
\times r$ matrix whose $j$-th column are the coordinates of $x_j$ is
the working basis.

\begin{algorithm}[h]\label{algo:modmult}
\KwIn{$A_1, A_2 \in L[X,\sigma]$, $\lambda \in L^\times$}
\KwOut{$A = A_1A_2 \pmod {X^r {-}a}$ where $a = N_{L/K}(\lambda)$}
$a = N_{L/K}(\lambda)$\\
$\{b_0,\dots, b_{r-1}\} = \texttt{NormalBasis}(L/K)$\\
$\tilde b_{r-1} = b_{r-1}$\\
\For{$r{-}1 \geq i \geq 1$}{
$\tilde b_{i-1} = a \sigma(\tilde b_i)$}
$P = \matw(\tilde b_0, \dots, \tilde b_{r-1})$\\
$B = \sum_{i=0}^{r-1} \tilde b_i T^i$\\
\For{$1 \leq i \leq 2$}{
$C_i = A_iB \pmod {T^r {-} a}$, \text{write }$ C_i = \sum_{i=0}^{r-1} c_{i,j}T^j$\\
$N_i = \matw(c_{i,0}, \dots, c_{i,r-1})$
}
$N = N_1P N_2$\\
$C = (\beta_0~\dots~ \beta_{r-1}) N \begin{pmatrix} 1\\ T \\ \vdots \\ T^{r-1}\end{pmatrix}$\\
$A = CB^{-1} \pmod {T^r {-} a}$\\
\Return{$A(X)$}
\caption{\texttt{ModMult}}
\end{algorithm}

\begin{prop}
Algorithm \textrm{\tt ModMult} computes the product $A_1A_2$ in $L[X, \sigma]/(X^r{-}a)$ in $O(r^\omega)$ operations in $K$.
\end{prop}
\begin{proof}
Multiplication of polynomials in $L$ modulo $T^r {-} a$ requires $\tilde O(r^2)$ operations in $K$. Multiplication of matrices of size $r$ in $K$ requires $O(r^\omega)$ operations in $K$. Hence the global complexity is $O(r^\omega)$ operations in $K$.
\end{proof}

\subsubsection{Multiplication modulo $Z(X^r)$} 
\label{subs:modmultZ}

Let $K'/K$ be a finite extension.
Define $L' = K' \otimes L$; it is an {\'e}tale $K'$-algebra endowed with the endomorphism $\sigma' = \text{id} \otimes \sigma$ that extends $\sigma$ and has order $r$. 

\begin{rmq}
The algebra $L'$ is not necessarily a field (for instance, when $K' = L$,
it splits as a product $L^r$). It is the reason why we needed to place
this paper in the more general setting of \'etale algebras.
\end{rmq}

Let $\lambda \in (L')^{\times}$. Set $a = N_{L'/K'}(\lambda) = \lambda\sigma(\lambda)\cdots\sigma^{r-1}(\lambda) \in K'$. We assume that
$K' = K(a)$. Let $Z \in K[T]$ be the minimal polynomial of $a$. We want to generalize the results of \S \ref{subs:fastmod} to multiplication modulo $Z(X^r)$ (in \S \ref{subs:fastmod}, we have $K' = K$, $L' = L$ and $\lambda \in L^\times$). Note that if $(b_0, \dots, b_{r-1})$ is a normal basis of $L/K$, then $(1 \otimes b_0, \dots, 1 \otimes b_{r-1})$ is a normal basis of $L'/K'$.
\begin{lem}\label{lem:extension}
The canonical morphism $1 \otimes \textrm{id}$ : $L[X,\sigma] \rightarrow L'[X,\sigma]$ induces an isomorphism
$$L[X,\sigma]/Z(X^r)\simeq  L'[X,\sigma'] /(X^r {-} a).$$
\end{lem}
\begin{proof}
First note that $(X^r {-} a)$ is a two-sided ideal of $L'[X,\sigma] $, and that the canonical morphism $L[X,\sigma] \rightarrow L'[X,\sigma]$ induces a morphism $L[X,\sigma] \rightarrow L'[X,\sigma]/(X^r {-} a)$ which maps $X^r$ to $a$, hence the latter surjective. Moreover, by $K$-linearity, $Z(X^r)$ lies in the kernel of this map. 
We then get a surjective morphism of $K$-algebras $L[X,\sigma]/Z(X^r) \rightarrow L'[X,\sigma]/(X^r {-} a)$. Since both sides have dimension $r^2 \deg Z$ over $K$, this morphism is an isomorphism.
\end{proof}
We are now back exactly in the situation of Section \ref{subs:fastmod}, where $K$ has been replaced by $K'$ and $L$ by $L'$: all the computations can be carried out the same way, and passing back through the isomorphism of Lemma \ref{lem:extension}, we can perform fast multiplication modulo $Z(X^r)$. The algorithm is as follows:
\begin{algorithm}\label{algo:modmultZ}
\KwIn{$A_1, A_2 \in L[X,\sigma]$, $K'/K$ a finite extension, $\lambda \in L' = K' \otimes L$ nonzero, $a = N_{L'/K'}(\lambda) \in K'$ such that $K' = K(a)$, $Z \in K[T]$ the minimal polynomial of $a$ over $K$. }
\KwOut{$A = A_1A_2 \pmod {Z(X^r)}$ where $Z$ is the minimal polynomial of $a = N_{L'/K'}(\lambda)$ over $K$.}
Write $A_1 = \sum_{i=0}^{r-1} \alpha_i(X^r)X^i$, $A_2 = \sum_{i=0}^{r-1} \beta_i(X^r)X^i$\\
Let $\tilde A_1 = \sum_{i=0}^{r-1} \alpha_i(a)X^i$, $\tilde A_2 = \sum_{i=0}^{r-1} \beta_i(a)X^i$\\
Compute $\tilde A = \tilde A_1 \tilde A_2$ using \texttt{ModMult} in $L'[X, \sigma]/(X^r {-} a)$ endowed with the normal basis $(1 \otimes b_i)$\\
Write $A = \sum_{i=0}^r \gamma_i(a) X^i$\\
\Return{$A = \sum_{i=0}^r \gamma_i(X^r) X^i$}
\caption{\texttt{ModMultZ}}
\end{algorithm}

\begin{prop}
Algorithm \ref{algo:modmultZ} computes the product $A_1A_2$ in 
$L[X,\sigma]/(Z(X^r))$ with $O(r^{\omega}\deg Z)$ operations in $K$.
\end{prop}

\subsection{Reconstruction with CRT}

Let $A_1, A_2 \in L[X,\sigma]$ be two skew polynomials. We recall that 
our aim is to design a fast algorithm for computing the product $P = A_1 
A_2$. We set $d = \deg P$.

\medskip

\noindent
{\bf Multiplication in large degree.}
We first assume that the polynomial $P = A_1 A_2$ has degree larger than 
$r$. In this case, the idea is to evaluate the $P$ modulo various 
$Z_i(X^r)$ using Algorithm \texttt{ModMultZ} and then to reconstruct the 
result using a non commutative version of the Chinese Remainder Theorem. 
The precise result we need is given by the following Proposition.

\begin{prop}
\label{prop:skew_crt}
Let $Z_1, \dots, Z_m \in K[T]$ be pairwise coprime polynomials, and let 
$Z = Z_1\cdots Z_m$. Then the natural map:
$$L[X,\sigma]/Z(X^r) \rightarrow 
L[X,\sigma]/Z_1(X^r) \times \cdots \times L[X,\sigma]/Z_m(X^r)$$
is an isomorphism of $K$-algebras.
\end{prop}

\begin{proof}
Since the domain and the codomain have the same dimension over $K$,
it is enough to prove the surjectivity.
For $i$ between $1$ and $m$, consider $A_i \in L[X,\sigma]/Z_i(X^r)$
and write it:
$$A_i = A_i^{(0)}(X^r) + A_i^{(1)}(X^r) X + \cdots + 
A_i^{(r{-}1)}(X^r) X^{r-1}$$
where the $A_i^{(j)}$'s are polynomials with coefficients in $L$.
For a fixed $j \in \{0, \ldots r{-}1\}$, let $A^{(j)} \in L[T]$ be
a polynomial such that the congruence $A^{(j)} \equiv A_i^{(j)} \pmod 
{Z_i}$ holds in the \emph{commutative} ring $L(T)$. We can therefore
write $A^{(j)} = A_i^{(j)} + Z_i Q_i^{(j)}$
for some polynomials $Q_i^{(j)} \in L[T]$. Noting that the inclusion 
$L[T] \to L[X,\sigma]$, $T \mapsto X^r$ is a ring homomorphism 
(\emph{i.e.} the multiplication on $L[T]$ agrees with that on 
$L[X,\sigma]$), we deduce that the equality
$$A^{(j)}(X^r) = A_i^{(j)}(X^r) + Z_i(X^r) \cdot Q_i^{(j)}(X^r)$$
holds in $L[X,\sigma]$. Multiplying it by $X^j$ on the right and
summing up over $j$, we end up with $A \equiv A_i \pmod{Z_i}$ for
all $i$. Surjectivity is proved.
\end{proof}

\begin{rmq}
\label{rmq:skew_crt}
The above proof is constructive. More precisely it shows that solving 
the Chinese Remainder problem of degree $d$ in $L[X,\sigma]$ with 
\emph{central moduli} reduces to solving $r$ independant Chinese 
Remainder problems of degree $\frac d r$ in the \emph{commutative} 
ring $L[X^r]$ and therefore can be achieved for a cost of $\tilde 
O(d)$ operations in $L$, corresponding to $\tilde O(dr)$ operations
in $K$ (see~\cite{GaGe03}, \S 10.3).
\end{rmq}

It remains now to explain how the moduli $Z_i(X^r)$'s can be 
constructed. We will do it in two different concrete contexts: first, 
the case of finite fields and second, the case of number fields.

\medskip

\noindent
{\it The case of finite fields.}
We assume that $K$ and $L$ are finite fields and write $q$ for the
cardinality of $K$. We consider an auxiliary finite extension $K'$ of
$K$ of degree $n$ and build the compositum $L' = K' \otimes_K L$. We
endow $L'$ with the uniform measure. We assume that $n$ is chosen
sufficiently large so that:
\begin{equation}
\label{eq:largen}
q^n \geq \max(64n, 8r).
\end{equation}
Asymptotically the latest condition is fulfiled as soon as $n$ grows
at least as fast as $\log r$.

\begin{lem}
\label{lem:proba}
Let $t$ be an integer such that $4 t^2 \leq n q^n$.
Let $\lambda'_1, \ldots, \lambda'_t$ be random independant elements of
$L'$. 
Then the $N_{L'/K'}(\lambda'_i)$'s all generate $K'$ over $K$
and are pairwise non-conjugate over $K$ with probability at least
$\frac 1 2$.
\end{lem}

\begin{proof}
The \'etale algebra $L'$ splits as a product $(M')^g$ where $M'$ is a 
finite extension of $K'$ of degree $f$ and $g$ is a positive integer.
Moreover if $x \in L'$ decomposes as $x = (x_1, \ldots, x_g)$, we
have:
$$N_{L'/K'}(x) = N_{M'/K'}(x_1) \cdots N_{M'/K'}(x_g).$$
Observe that the norm map $N_{M'/K'}$ takes the value $0$ only at $0$. 
Hence the probability that $N_{M'/K'}$ vanishes is $q^{-nf}$. Therefore 
$N_{L'/K'}$ vanishes with probability $1 - (1 - q^{-nf})^g$.
As for the nonzero values of $K'$, they are reached by $N_{L'/K'}$ with 
uniform probability because $N_{L'/K'}$ is a surjective group 
homomorphism, \emph{i.e.}
$$\text{Prob}
\big[N_{L'/K'} = a \big] = \left(1 - \frac 1{q^{nf}}\right)^g 
\cdot \frac 1 {q^n - 1}$$
for all $a \in K'$, $a \neq 0$. Let $c_n$ be the number of elements of
$K'$ that generate $K'$ over $K$. 
The probability that a fixed $\lambda'_i$ satisfy the requirement 
$K\big(N_{L'/K'}(\lambda'_i)\big) = K'$ is then $(1 - q^{-nf})^g \cdot 
\frac {c_n} {q^n - 1}$. Assuming that this occurs, the probability that 
the $N_{L'/K'}(\lambda'_i)$'s are pairwise non-conjugate is:
$$\left(1 - \frac 1{nc_n}\right) \cdot
\left(1 - \frac 2{nc_n}\right) \cdots
\left(1 - \frac {t-1}{nc_n}\right).$$
Putting all together, we find the probability of success:
$$\left(1 - \frac 1{q^{nf}}\right)^g \cdot \frac {c_n}{q^n-1} \cdot
\left(1 - \frac 1{c_n}\right) \cdots
\left(1 - \frac {t-1}{nc_n}\right)$$
which is at least:
\begin{equation}
\label{eq:probasuccess}
\frac{c_n}{q^n{-}1} \left(1 - \frac g{q^{nf}} - 
\frac{t(t{-}1)}{2\:c_n}\right)
\geq \frac{c_n}{q^n} - \frac r{q^n} - \frac{t(t{-}1)}{2n\:q^n}.
\end{equation}
Clearly $q^n{-}c_n$ is the cardinality of the union of all strict 
subextensions of $K'$. Therefore:
$$q^n - c_n \leq \sum_{m | n, m < n} q^m \leq 2 \sqrt n \cdot q^{n/2}$$
the latter inequality coming from the fact that $n$ has at most 
$2 \sqrt n$ divisors. From \eqref{eq:largen}, we derive $q^n{-}c_n \leq 
\frac {q^n} 4$. On the other hand, it follows from our assumptions that
$r \leq \frac {q^n} 8$ and $\frac{t(t{-}1)}{2n} \leq \frac{t^2}{2n} \leq
\frac {q^n} 8$. Combining with \eqref{eq:probasuccess}, we find that the
probability of success is at least $\frac 1 2$.
\end{proof}

\begin{algorithm}[h]\label{algo:globalmult}
\KwIn{$A_1, A_2 \in L[X,\sigma]$ of degree $\leq d$}
\KwOut{$P = A_1A_2$}
Choose $n$ and $K'$ such that Eq.~\eqref{eq:largen} holds and \newline
\phantom{}\hfill
$\displaystyle
\frac{8d}{nr} \cdot \left(\frac{2d}{nr} + 1\right) \leq n q^n$\hfill\null
\\
Set $t = \lceil\frac{2d}{nr}\rceil$\\
Pick $\lambda'_1, \ldots, \lambda'_t \in L' = K' \otimes_K L$ at random\\
\For{$1 \leq i \leq t$}
{Compute the min. poly. $Z_i \in K[T]$ of $N_{L'/K'}(\lambda'_i)$\\
Compute $P_i = A_1A_2 \in L[X,\sigma]/Z_i(X^r)$ 
\newline\phantom{}\hfill /\!\!/ \textit{use Algorithm \textrm{\tt ModMultZ}}}{}
Compute $P$ such that $\deg A \leq 2d$ and 
$P \equiv P_i \pmod{Z_i}$\label{line:crt}
\newline\phantom{}\hfill /\!\!/ \textit{use Proposition \ref{prop:skew_crt}}\\
\Return{$P$}
\caption{\texttt{Mult}}
\end{algorithm}

\begin{thm}
Let $A_1, A_2 \in L[X,\sigma]$ of degree $d \geq r$. 
Then Algorithm \textrm{\tt Mult} computes the product $A_1A_2$ 
within $O(dr^{\omega-1})$ operations in $K$ with probability of 
success at least $\frac 1 2$.
\end{thm}

\begin{proof}
Observe first that $n$ can be chosen such that $n = O(\log d + \log r)$.
Computing the product in $L[X,\sigma]/Z_i(X^r)$ requires $O(r^\omega n)
= \tilde O (r^\omega)$ 
operations in $K$. Moreover by Remark~\ref{rmq:skew_crt}, the 
reconstruction (line \ref{line:crt}) can be done for a cost of $\tilde O(rd)$ 
operations in $K$. The overall cost of \texttt{Mult} is 
then $\tilde O(dr^{\omega-1})$ as announced.
The fact that the probability of success is at least $\frac 1 2$
follows from Lemma \ref{lem:proba}.
\end{proof}

\noindent
{\it The case of number fields.}
We assume that $K$ and $L$ are number fields. It is then known that the 
image of the norm map $N_{L/K} : L^\star \to K^\star$ has index $r$ in 
$K^\star$. More precisely, class field theory teaches us that $K^\star / 
N_{L/K}(L^\star)$ is canonically isomorphic to the Galois group of the
abelian extension $L/K$, \emph{i.e.} to $\Z/r\Z$. In particular, the
image of $N_{L/K}$ is infinite meaning that if we take a finite set of 
random elements $\lambda \in L$, it is likely that the norm of the 
$\lambda$'s will be pairwise distinct.
We can then reapply the strategy used in the case of finite field 
without having to work with an auxiliary extension $K'$. We end up this 
way with a probabilistic Las Vegas algorithm whose complexity is $\tilde 
O(d r^{\omega{-}1})$ operations in $K$ and whose probability of success 
is high.

\medskip

\noindent
{\bf Multiplication in small degree.}
The idea for fast multiplication in small degree is that if a skew 
polynomial has degree $d \ll r$, it is determined by its values on 
$d{+}1$ linearly independent elements of $L$. Hence, starting with two 
skew polynomials $A_1,A_2$ whose degrees add up to $d$, we should be 
able to compute their product by composing of two $K$-linear maps over 
vector spaces of dimension $d{+}1$.
However, we know some efficient algorithm for evaluating $A(\sigma)$
only on a subspace of $L$ which is spanned by the first vectors of a
normal basis. For this reason, it order to compute $A_1A_2(b_0), \ldots, 
A_1A_2(b_{d-1})$, we shall need to know the whole of the linear map 
$A_1(\sigma)$ (because $A_2(b_0), \ldots, A_2(b_{d-1})$ are in general
nothing to do with a truncated normal basis).

\begin{algorithm}[!h]\label{algo:smalldegmult}
\KwIn{$A_1, A_2 \in L[X,\sigma]$, $\deg A_1 + \deg A_2 < r$}
\KwOut{$P = A_1A_2$}
Set $d = \deg A_1 + \deg A_2$\\
Compute $A_2(b_0), \ldots, A_2(b_d)$
\hfill /\!\!/ \textit{use Corollary~\ref{cor:eval_small}}\\
Compute the matrix of $P(\sigma)$
\hfill /\!\!/ \textit{use Proposition~\ref{prop:evalbasis}}\\
Compute $c_0 = A_1A_2(b_0), \ldots, c_d =A_1A_2(b_d)$
\newline\phantom{}\hfill /\!\!/ \textit{matrix multiplication of sizes
$r \times r$ by $r \times (d{+}1)$}\\
Compute $P \in L[X,\sigma]$ s.t. $P(b_i) = c_i$ and $\deg P \leq d$.
\newline\phantom{}\hfill /\!\!/ \textit{use Algorithm \textrm{\tt 
SmallDegreeInterpolation}}\\
\Return{P}
\caption{\texttt{SmallDegreeMultiplication}}
\end{algorithm}

The complexity of the above algorithm is given by the next Theorem whose 
proof is straightforward after what we have already done (the bottleneck 
comes from the matrix multiplication step).

\begin{thm}
Let $A_1, A_2$ such that $\deg A_1 + \deg A_2 \leq d<r$. Then Algorithm 
\ref{algo:smalldegmult} computes the product $A_1A_2$ with 
$O(d^{\omega-2}r^2)$ operations in $K$.
\end{thm}

\noindent
{\bf Conclusion.}
As a conclusion, several algorithms with different complexities are 
available for the multiplication of skew polynomials. Precisely, we have 
designed in this paper one algorithm of complexity $\tilde O(d 
r^{\omega-1})$ when $d \geq r$ and an another algorithm of complexity 
$\tilde O(d^{\omega-2} r^2)$ when $d \leq r$.
Apart from that, Wachter-Zeh's algorithm \cite{PuWa16} performs the same
computation with complexity $\tilde O(d^{(\omega+1)/2} r)$ without any
assumption on $d$. The corresponding complexity curves are represented
on Figure~\ref{fig:complexity}.
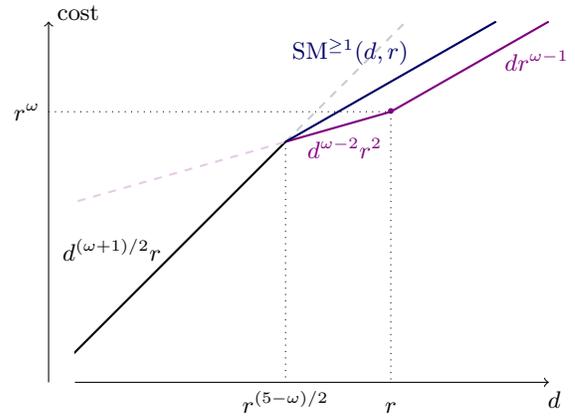
\begin{figure}
\hfill
\begin{tikzpicture}[xscale=7,yscale=4]
\begin{scope}
\clip (0.25,1.4) rectangle (1.37,2.9);
\draw[->] (0.4,1.6)--(1.3,1.6);
\draw[->] (0.35,1.6)--(0.35,2.8);
\node[below] at (1.31,1.6) { $d$ };
\node[right] at (0.35,2.83) { cost };
\end{scope}
\begin{scope}[thick]
\clip (0.4,1.5) rectangle (1.3,2.8);
\draw (0,1)--(0.8,2.4);
\draw[black!20, dashed] (0.8,2.4)--(2,4.5);
\draw[violet!20, dashed] (0,2)--(0.8,2.4);
\draw[violet] (0.8,2.4)--(1,2.5)--(1.3,2.8);
\node[violet,scale=0.6] at (1,2.5) { $\bullet$ };
\draw[bleufonce] (0.8,2.4)--(1.2,2.8);
\end{scope}
\node[left] at (0.58,2.04) { $d^{(\omega+1)/2} r$ };
\draw[dotted] (0.35,2.5)--(1,2.5)--(1,1.6);
\draw[dotted] (0.8,2.4)--(0.8,1.6);
\node[below] at (1,1.6) { $\vphantom{r^{(5-\omega)/2}}r$ };
\node[below] at (0.8,1.6) { $r^{(5-\omega)/2}$ };
\node[left] at (0.35,2.5) { $r^\omega$ };
\node[violet,left] at (1,2.38) { $d^{\omega-2} r^2$ };
\node[violet,right] at (1.2,2.67) { $d r^{\omega-1}$ };
\node[bleufonce,left] at (1.05,2.7) { $\SM^{\geq 1}(d,r)$ };
\end{tikzpicture}
\hfill\null

\vspace{-5mm}

\caption{Complexity profiles (log-log scale)}
\label{fig:complexity}
\end{figure}
Putting all together, we find that the product in $L[X,\sigma]$ can be 
performed within $\tilde O(\SM(d,r))$ operations in $K$ where: 
$$\begin{array}{r@{\hspace{0.5ex}}ll}
\SM(d,r) & = d^{(\omega+1)/2} r & \text{for } d \leq r^{(5-\omega)/2} 
\smallskip \\
& = d^{\omega-2} r^2 & \text{for } r^{(5-\omega)/2} \leq d \leq r 
\smallskip \\
& = d r^{\omega-1} & \text{for } d \geq r.
\end{array}$$
As already discussed in the introduction, we expect to lower the
complexity to $\tilde O(d^{\omega-1} r)$ in the range $d \leq r$ and, 
until now, we have not succeeded in doing so.

\section{Other operations\\and applications}

Classically, fast multiplication algorithms can be used to speed up many 
other computations. This general philosophy works for skew polynomials
as well and was concretized in \cite{CaLe17}, \S 3.2.
Below, we analyze briefly the impact of the algorithms designed above in 
this paper.

In order to state our complexity results more elegantly, we introduce 
the function $\SM^{\geq 1}$ defined by:
$$\SM^{\geq 1}(d,r) = \sup_{d' \leq d} \left(
\SM(d',r) \cdot \frac d{d'}\right).$$
A direct computation shows that:
$$\begin{array}{r@{\hspace{0.5ex}}ll}
\SM^{\geq 1}(d,r) & = 
d^{(\omega+1)/2} r & \text{for } d \leq r^{(5-\omega)/2} 
\smallskip \\
& = d r^{4/(5-\omega)} & \text{for } d \geq r^{(5-\omega)/2}.
\end{array}$$
The function $\SM^{\geq 1}$ (viewed as a function of the variable $d$) 
is the smallest function above $\SM$ whose ``log-log slope'' is always 
at least $1$ (see Figure~\ref{fig:complexity}). The notation comes from 
this interpretation.

With $\omega = 2.37$, we have 
$\SM^{\geq 1}(d,r) \approx d^{1.69}r$ for $d \leq r^{0.76}$
$\SM^{\geq 1}(d,r) \approx dr^{1.52}$ for larger $d$.

\medskip

\noindent
{\bf Euclidean division.}
An algorithm that performs (right) Euclidean divisions in $L[X,\sigma]$ 
and takes advantage of fast multiplication algorithm is depicted in 
\cite{CaLe17}, \S 3.2.1 (Algorithm \texttt{REuclideanDivision}). Proposition 
3.2.3 of \emph{loc. cit.} extends readily to the settings of this paper 
and shows that the aforementioned algorithm has a complexity cost of 
$\tilde O(\SM^{\geq 1}(d,r))$ operations in $K$.

\medskip

\noindent
{\bf {\sc gcd} and {\sc lcm} computation.}
The classical half-gcd algorithm that we already mentioned above (see 
\S \ref{subs:interp_small} and \cite{GaGe03}, \S 11) works in 
the same way to compute left and right {\sc gcd}'s of skew polynomials.
The precision corresponding algorithm is written in \cite{CaLe17},
\S 3.2.2 (Algorithm \texttt{FastExtendedRGCD}).

\begin{prop}
The algorithm \texttt{FastExtendedRGCD} of \cite{CaLe17}, \S 3.2.2 (using
fast multiplication algorithms described above in this paper as
primitives)
runs in $\tilde O(\SM^{\geq 1}(d,r))$ operations in $K$.
\end{prop}

\begin{proof}
A careful look at the algorithm \texttt{FastExtendedRGCD} shows that its 
complexity in operations in $K$ is bounded by $T(d,r)$ where $T(d,r)$ 
satisfies the recurrence relation:
$$\textstyle
T(d,r) \leq 2 \: T\big(\frac d 2,r\big) + \tilde O \big(\SM(\frac d 2,r)\big).$$
By induction, it follows that for $m \geq 0$, 
\begin{align*}
T(d,r) 
 & \leq 2^m T\big({\textstyle \frac d{2^m}},r\big) + \tilde O\left(\sum_{j=1}^{m} 2^j \SM\big({\textstyle \frac d{2^j}},r\big)\right) \\
 & \leq 2^m T\big({\textstyle \frac d{2^m}},r\big) + \tilde O\left(m \cdot\SM^{\geq 1}(d,r)\right).
\end{align*}
Taking $m = \lfloor \log_2 d \rfloor$, we get
$T(d,r) = \tilde O (\SM^{\geq 1}(d,r))$ as expected.
\end{proof}

\begin{rmq}
A similar complexity is available for the computation of {\sc lcm}'s.
\end{rmq}

\noindent
{\bf Minimal subspace polynomial.}
Let $(x_1, \ldots, x_d)$ be a family of elements of $L$ which is free 
over $K$. We are interesting in computing the unique monic polynomial $P 
\in L[X,\sigma]$ of degree $d$ such that $P(x_i) = 0$ for all $i \in 
\{1, \ldots, d\}$. 

\begin{lem}
\label{lem:evalrem}
For $x \in L$, $x \neq 0$, the value $\frac{P(x)} x$ is the remainder 
in the right
Euclidean division of $P$ by $X - \frac{\sigma(x_i)}{x_i}$.
\end{lem}

\begin{proof}
It is a direct computation.
\end{proof}

Lemma \ref{lem:evalrem} shows that the polynomial $P$ we are looking for 
is nothing but the left-lcm of the polynomials $X - \frac{\sigma(x_i)}
{x_i}$. As a consequence, $P$ can be computed for a cost of
$\tilde O (\SM^{\geq 1}(d,r))$ operations in $K$ using fast algorithms
for {\sc lcm} computation together with a ``tree division strategy'' 
\cite{GaGe03}, \S 10.1.

\medskip

\noindent
{\bf General multievaluation.}
We consider again a free family $(x_1, \ldots, x_d)$ of elements of $K$.
The general multievaluation problem consists in evaluating a given 
polynomial $P \in K[X,\sigma]$ of degree $d$ at the $x_i$'s. Thanks
to Lemma~\ref{lem:evalrem}, the value $P(x_i)$ agrees with $x_i$ times
the remainder of the right division of $P$ by 
We are then reduced to compute the reduction of a given polynomials
modulo some given moduli. This can be done efficiently using the
strategy of \cite{GaGe03}, \S 10.1 for a cost of $\tilde O(\SM^{\geq 1}(d,r))$
operations in $K$. If $d$ are $r$ have the same order of magnitude,
one can preferably compute the matrix of $P(\sigma)$ using the
formula of Proposition \ref{prop:evalbasis} and derive from it 
the values of the $P(x_i)$'s thanks to a single matrix multiplication.
The cost of the resulting algorithm is $O(r^\omega)$.

\begin{rmq}
If the $x_i$'s are the first vectors of a normal basis of 
$L$ over $K$, one can use directly the algorithm of \S \ref{subs:evalincnorm}
which has a better complexity.
\end{rmq}

\noindent
{\bf General interpolation.}
We keep the family $(x_1, \ldots, x_d)$ and consider in addition some 
values $y_1, \ldots, y_d \in L$. We address the question of computing a 
polynomial $P$ of degree at most $d{-}1$ such that $P(x_i) = y_i$ for 
all $i$. Thanks to Lemma~\ref{lem:evalrem}, the above problem reduces to 
solve the following Chinese Remainder system:
$$P(x_i) \equiv x_i y_i \pmod{\textstyle
X - \frac{\sigma(x_i)}{x_i}}$$
which again can be done for a cost of $\tilde O(\SM^{\geq 1}(d,r))$
operations in $K$.

\begin{rmq}
If the $x_i$'s are the first vectors of a normal basis of 
$L$ over $K$, one can use directly the Algorithm
\texttt{SmallDegreeInterpolation} which has a better complexity.
\end{rmq}

\noindent
{\bf Gabulin codes.}
The solution sketched above to the general multievaluation problem 
allows us to encode messages in the framework of (generalized) Gabidulin 
codes \cite{Ro15} in complexity $O(n^\omega)$ where $n$ is the length of
the code. (Better complexities are possible when the dimension of the
code is much smaller than its length.)
In the similar fashion, efficient decoding is also possible using the 
key equation together with the half-{\sc gcd} algorithm. The
resulting algorithms run in $\tilde O(SM^{\geq 1}(n,k))$ operations
in $K$ where $n$ and $k$ denotes the length and the dimension of the
Gabidulin code respectively.

\end{document}